\documentclass[12pt,onecolumn,draftcls]{IEEEtran}

\usepackage{subfig}
\usepackage{graphicx}
\usepackage{amsthm,amsfonts,amssymb,amsmath}
\usepackage{algorithm}
\usepackage{algpseudocode}
\usepackage{color}
\usepackage{bm,bbm}

\newtheorem{theorem}{Theorem}
\newtheorem{remark}{Remark}
\newtheorem{lemma}{Lemma}

\newtheorem{corollary}{Corollary}
\newcommand{\Z}{\mathbb{Z}}

\newcommand{\R}{\mathbb{R}}

\newcommand{\lp}[1]{\mbox{LPL}(n,#1)}

\begin{document}

\title{The set of dimensions for which there are no linear perfect $2$-error-correcting Lee codes has positive density}
\author{Claudio~Qureshi}
\author{
    \IEEEauthorblockN{Claudio~Qureshi\IEEEauthorrefmark{1}}\\
    {\small
    \IEEEauthorblockA{\IEEEauthorrefmark{1}Institute of Mathematics, Statistics and Scientific Computing, University of Campinas, Brazil. 
    \\cqureshi@ime.unicamp.br} \\
}}

\maketitle

\vspace{-1.4cm}
\begin{abstract}
The Golomb-Welch conjecture states that there are no perfect $e$-error-correcting Lee codes in $\Z^n$ ($\mbox{PL}(n,e)$-codes) whenever $n\geq 3$ and $e\geq 2$. A special case of this conjecture is when $e=2$. In a recent paper of A. Campello, S. Costa and the author of this paper, it is proved that the set $\mathcal{N}$ of dimensions $n\geq 3$ for which there are no linear $\mbox{PL}(n,2)$-codes is infinite and $\#\{n \in \mathcal{N}: n\leq x\} \geq \frac{x}{3\ln(x)/2} (1+o(1))$. In this paper we present a simple and elementary argument which allows to improve the above result to $\#\{n \in \mathcal{N}: n\leq x\} \geq \frac{4x}{25} (1+o(1))$. In particular, this implies that the set $\mathcal{N}$ has positive (lower) density in $\Z^+$.
\end{abstract}

\section{Introduction and Preliminaries}\label{SectionIntro}

The Lee metric was introduced in 1958 by C. Y. Lee for signal transition over certain noisy channels \cite{Lee58}. This is one of the most important metrics considered in error-correcting codes due to several applications such as constrained and partial-response channels \cite{RS94}, flash memory \cite{Schwartz12}, interleaving schemes \cite{BBV98}, multidimensional burst-error-correction \cite{EY09}, among others.

Let $x=(x_1,\ldots, x_n)$ and $y=(y_1,\ldots, y_n)$ two words of $\Z_q^n$ ($q\geq 2$) or $\Z^n$. The Lee metric is defined as follows.
$$d(x,y)=\left\{ \begin{array}{ll}
\sum_{i=1}^n \min\left(|x_i-y_i|, m-|x_i-y_i| \right) & \textrm{if }x,y \in \Z_q^n, \\
\sum_{i=1}^n |x_i-y_i| & \textrm{if }x,y \in \Z^n.
\end{array} \right.$$

There is a correspondence between $e$-error-correcting Lee codes in $\Z_q^n$ and $e$-error-correcting Lee codes in $\Z^n$ whenever $q\geq 2e+1$ (this case is know as "large alphabet") via the so called Construction A \cite{CS13}. This correspondence preserve perfect Lee codes. In the seminal paper of S. W. Golomb and L. R. Welch \cite{GW70} the authors studied perfect Lee codes in $\Z^n$. They presented several results including the construction of $\mbox{PL}(2,e)$-codes for every $e\geq 1$ and the construction of $\mbox{PL}(n,1)$-codes for every $n\geq 1$. In this paper the authors raised their famous conjecture (known as the Golomb-Welch conjecture) that for $e\geq 2$ and $n\geq 3$ there are no $\mbox{PL}(n,e)$-codes.

This conjecture is one of the most important open problems in the area of perfect Lee codes. In spite of great effort and several papers with partial results towards the conjecture, it is believed that the Golomb-Welch conjecture is far to being solved. A recent survey on this conjecture is given in \cite{HK18}.

This paper focus on the case $e=2$ of the Golomb-Welch conjecture, restricted to linear codes. This case is believed to be the most difficult case as pointed out in \cite{HK18}. For dimensions $n=3,4,5$ the non-existence of $\mbox{PL}(n,e)$-codes for every $e\geq 2$ can be proved from the non-existence of $\mbox{PL}(n,2)$-codes \cite{Horak09a}. Horak and Grosek proved in \cite{HG14} the non-existence of linear $\mbox{PL}(n,2)$-codes for dimensions $n\leq 12$. Kim introduce a new strategy in \cite{Kim17} and prove the non-existence of $\mbox{PL}(n,2)$-codes for some values of $n$ which is expected to be infinite, provided that $2n^2 +2n+1$ is a prime satisfying certain conditions. Let $\mathcal{N}$ denote the set of dimensions $n\geq 3$ such that there are no linear $\mbox{PL}(n,2)$-codes. In \cite{QCC18}, the authors prove that the set $\mathcal{N}$ is infinite via the existence of infinitely many friendly primes. The proof is long but mostly elementary and follows some ideas introduced in \cite{Kim17}, however the proof of the existence of the  infinitely many friendly primes is based on a (non-elementary) very strong result of Wiertelak \cite{Wiertelak78} on the distribution of certain classes of primes. The authors of \cite{QCC18} also obtain an estimative for the size of $\mathcal{N}$: $\#\{n \in \mathcal{N}: n\leq x\} \geq \frac{x}{3\ln(x)/2} (1+o(1))$. This estimative is enough to prove that $\mathcal{N}$ has infinitely many elements but it is not enough to prove that this set has positive lower density. We recall that the lower density (in this paper we refer it simply as \emph{density}) of a subset $\mathcal{N}\subseteq \Z^+$ is given by $d(\mathcal{N}) = \liminf_{n\to \infty} \#\{n \in \mathcal{N}: n\leq x\}/x$. In this paper we show a simple elementary argument based on congruences to prove not only that $\mathcal{N}$ has infinitely many elements but also that $\mathcal{N}$ has positive density, greater or equal than $4/25$.

\section{The main result}

We denote by $\lp{e}$ the set of all linear $\mbox{PL}(n,e)$-codes and by $B^n(e)$ the $n$-dimensional Lee ball with radius $e$ and centered at the origin, that is $B^n(e)=\{x\in\Z^n: d(x,0)\leq e\}$. The cardinality of the ball is given by $\#B^n(e)= \sum_{i=0}^{\min\{n,e\}} 2^i \binom{n}{i} \binom{e}{i}$ \cite{GW70}. In particular $\#B^n(2)=2n^2+2n+1$. The proof of our non-existence theorem has three main ingredients. The first is a criterion of Horak and AlBdaiwi for the non-existence of perfect Lee codes.

\begin{lemma}[\cite{HA12a}, Theorem 6]\label{Lemma1}
$ \lp{e}\neq \emptyset$ if and only if there is an abelian group $G$ and a homomorphism $\phi : \Z^n \rightarrow G$ such that $\phi|_{B^{n}(e)}: B^n(e)\rightarrow G$ is a bijection.
\end{lemma}

The second ingredient is a simple result about abelian groups and epimorphisms (i.e. surjective homomorphism of groups).

\begin{lemma}\label{Lemma2}
If $G$ is an abelian group with order $|G|=mp$ with $p$ prime and $p \nmid m$ then there is an epimorphism $\psi: G \rightarrow \Z_p$. In particular $\psi$ is an $m$-to-$1$ map.
\end{lemma}

\begin{proof}
By the fundamental theorem of finite abelian groups there is a group isomorphism $\phi: G \rightarrow \Z_p \times G'$ where $G'$ is an abelian group of order $m$. Let $\pi: \Z_p \times G' \rightarrow \Z_p$ given by $\pi(x,y)=x$. Since both $\phi$ and $\pi$ are epimorphisms, their composition $\psi = \pi\circ \phi : G \rightarrow \Z_p$ is also an epimorphism. The groups $\Z_p$ and $G/\ker(\psi)$ are isomorphic (by the first isomorphism theorem for groups), thus they have the same cardinality. Hence $\#\ker(\psi)=\#G/ \#\Z_p = m$ and since $\psi$ is an epimorphism, it is an $m$-to-$1$ map. 
\end{proof}

The third ingredient is a formula of D. Kim (see the proof of Theorem 4 of \cite{Kim17}).

\begin{lemma}\label{Lemma3}
Let $p$ be a prime number, $k$ be a positive integer, $x=(x_1,\cdots, x_n)\in \Z_p^n$ and $Q_k(x)$ be given by
\begin{align*}
Q_k(x) :=& \sum_{i=1}^{n}\left((x_i)^{2k}+(-x_1)^{2k}+ (2x_i)^{2k} + (-2x_i)^{2k} \right)\\ &+ \sum_{1\leq i < j \leq n} \left( (x_i+x_j)^{2k} + (x_i-x_j)^{2k} +  (-x_i+x_j)^{2k} +  (-x_i-x_j)^{2k}  \right).
\end{align*}
Then $Q_k(x)$ can be expressed in terms of the power sums $S_{2t}(x)=\sum_{i=1}^{n}x_i^{2t}$ $1\leq t\leq k$ as
$$Q_k(x)=(4^{k}+4n+2)S_{2k} + 2 \sum_{t=1}^{k-1}\binom{2k}{2t} S_{2t}S_{2(k-t)}.$$
\end{lemma}

\begin{remark}\label{Remark1}
We will use the Kim's formula only with $k=1$ and $k=2$. In these cases we have:
\begin{equation}
\left\{ \begin{array}{l}
Q_1(x) = (4n+6) S_2(x)   \\
Q_2(x) = (4n+18) S_4(x) + 12\cdot S_2(x)^2
\end{array}    \right.
\end{equation}
\end{remark}

\begin{theorem}\label{TheoremMain}
$\lp{2}=\emptyset$ if $n\equiv 8, 13, 18, 23 \pmod{25}$.
\end{theorem}

\begin{proof}
We only prove the result for the case $n\equiv 8 \pmod{25}$ since the other cases are similar. By contradiction, suppose that $\lp{2}\neq \emptyset$ for some positive integer $n\equiv 8 \pmod{25}$. Let $B$ denote the Lee ball $B=B^{n}(2)$. By Lemma \ref{Lemma1} there is an abelian group $G$ and an homomorphism $\phi: \Z^n \rightarrow G$ such that the restricted map $\phi|_{B}: B \rightarrow G$ is a bijection. The order of $G$ is given by $|G|=\#B^{n}(2)=2n^2+2n+1\equiv 20 \pmod{25}$. Thus, we can write $|G|=5m$ with $5\nmid m$. By lemma \ref{Lemma2} there is an $m$-to-$1$ homomorphism $\psi:G \rightarrow \Z_{5}$. We define $\hat{\phi}=\psi\circ \phi: \Z^n \rightarrow \Z_5$ and $x_i:= \hat{\phi}(e_i) \in \Z_5$ for $1\leq i \leq n$, where $\{e_1,\ldots, e_n\}$ is the standard basis of $\R^n$. Since $\hat{\phi}|_{B}: B \rightarrow \Z_5$ is an $m$-to-$1$ map, the multiset $\{0\} \cup \{ x_i, 2x_i : 1\leq i \leq n  \} \cup \{\pm x_i \pm x_j: 1\leq i < j \leq n\}$ contains exactly $m$ times every element of $\Z_5$. By Lemma \ref{Lemma3}, Remark \ref{Remark1} and using that $n\equiv 3 \pmod{5}$  we have:
$$ \left\{\begin{array}{l}
38\cdot S_2(x) \equiv m\cdot (0^2+1^2+\cdots+4^2)\equiv 0 \pmod{5}\\
50\cdot S_4(x) + 12\cdot S_2(x)^2 \equiv m\cdot (0^4+1^4+\cdots+4^4)\equiv 4m \pmod{5}
\end{array}  \right.   $$
This is a contradiction since the first congruence implies $S_2(x)\equiv 0 \pmod{5}$ and the second congruence implies $S_2(x)\not\equiv 0 \pmod{5}$.
\end{proof}

\begin{corollary}
The set $\mathcal{N}:=\{n\geq 3: \textrm{there are no linear $\mbox{PL}(n,e)$-codes}\}$ has positive density $d(\mathcal{N})\geq 4/25$. In particular $\mathcal{N}$ has infinitely many elements.
\end{corollary}

\section*{Acknowledgement}

The author was supported by FAPESP under grants 2015/26420-1 and 2013/25977-7.

\bibliographystyle{alpha}

\end{document}